\documentclass[submission,copyright,creativecommons]{eptcs}
\providecommand{\event}{AFL 2017} 
\usepackage{underscore}           
\usepackage{amsmath,amsthm,amssymb,latexsym,comment}
\usepackage{enumerate}
\newcommand{\N}{\mathbb{N}}

\newcommand{\Z}{\mathbb{Z}}

\newtheorem{theorem}{Theorem}
\newtheorem{fact}{Fact}
\newtheorem{proposition}{Proposition}
\newtheorem{corollary}{Corollary}
\newtheorem{lemma}{Lemma}

\title{Generalized Results on Monoids as Memory}
\author{\"{O}zlem Salehi
	\institute{Bo\u{g}azi\c{c}i University \\Department of Computer Engineering\\ Bebek 34342 \.{I}stanbul, Turkey\thanks{The first author is partially
			supported by \emph{T\"{U}B\.{I}TAK} (Scientific and Technological Research Council of Turkey).}}
	\email{ozlem.salehi@boun.edu.tr}
	\and
	Flavio D'Alessandro 
	\institute{	Universit\`a di Roma ``La Sapienza''\\ Dipartimento di Matematica\\ Piazzale Aldo Moro 2, 00185 Roma, Italy}
	\institute{	Bo\u gazi\c ci University\\ Department of Mathematics\\ Bebek 34342, \.{I}stanbul, Turkey\thanks{The second author is supported by a EC-FP7 Marie Curie-\emph{T\"{U}B\.{I}TAK} Co-Funded Brain Circulation Scheme Project 2236 Fellowship.}}
	\email{dalessan@mat.uniroma1.it}
	\and
	A. C. Cem Say
	\institute{Bo\u{g}azi\c{c}i University\\ Department of Computer Engineering\\ Bebek 34342 \.{I}stanbul, Turkey}
	\email{say@boun.edu.tr}
}
\def\titlerunning{Generalized Results on Monoids as Memory}
\def\authorrunning{\"{O}. Salehi, F. D'Alessandro \& A. C. C. Say}

\begin{document}
	\maketitle
	
	\begin{abstract}
	   We show that some results from the theory of group automata and
		monoid automata still hold for more general classes of monoids and models. Extending previous work for finite automata over commutative groups, we prove that the context-free  language $\mathtt{L_1}^*=  \{a^n b^n : n\geq 1\}^* $ can not be recognized by any rational monoid automaton over a finitely generated permutable monoid. We show that the class of languages recognized by rational monoid automata over finitely generated completely simple or completely 0-simple permutable monoids is a semi-linear full trio. Furthermore, we investigate valence pushdown automata, and prove that they are only as powerful as (finite) valence automata. We observe that certain results proven for monoid automata can be easily lifted to the case of context-free valence
	grammars. 
		
%
	\end{abstract}

	\section{Introduction}
	
	A group automaton is a nondeterministic finite automaton equipped with a register which holds an element of a group. The register is initialized with the identity element of the group, and modified by applying the group operation at each step. An input string is accepted if the register is equal to the identity element at the end of the computation.
	This model has implicitly arisen under various names such as, for instance, nondeterministic blind counter machines \cite{Gr78}, and finite automata with multiplication \cite{ISK76}.  
	
	The notion of group automata has been actively investigated in the last decade, especially in the case of commutative groups and in the case of
	free (non commutative) groups, where remarkable results on the structure of the languages accepted by such automata have been proven  
	\cite{Co05,DM00,MS01}. 
	Subsequently, the notion of group automaton  has been  extended in at least two meaningful ways. The first one is that of monoid automaton, also known as valence automaton, by assuming that the register associated with the model is a monoid \cite{FS02,Ka06,Ka09}. The second one is that 
	of valence pushdown automaton introduced in \cite{FS02}, where the underlying model of computation is a pushdown automaton.

	The notions of monoid and group automata are also strictly related to 
	that of valence grammar introduced by P\v{a}un in \cite{Pa80}. A valence grammar is a formal grammar in which every rule
	of the grammar  is equipped with an element of a monoid called the valence of the rule. Words generated by the grammar are defined by successful derivations. A successful
	derivation is a derivation, that starts from the start symbol of the grammar and such
	that the product of the valences of its productions (taken in the obvious order) is the
	identity of the monoid. 
	
	
	In the case of context-free grammars, a thorough study of several remarkable structural properties of 
	the languages generated by the corresponding valence grammars has been done in \cite{FS02},
	over arbitrary monoids and in particular over commutative groups. 

In the case of monoid automata where the corresponding monoid is not a group, the requirement that a successful computation should end with the identity element limits the extent to which we can make use of the algebraic structure. In order to overcome this obstacle, a new model called monoid automaton with targets has been introduced and studied.
The targets of the automaton are two subsets of the monoid associated with the model, called respectively the initial set and the terminal set,
that define its successful computations. Precisely, at the beginning of the computation,
the register is initialized with an element from the initial set and a computation is successful if the register holds an element of the terminal set at the end of the computation. If the targets of the automaton are rational subsets of the monoid, the corresponding model is called rational monoid automaton.
The idea of allowing more general accepting configurations has been also applied to valence grammars, leading to the concept of valence grammars with target sets \cite{FS01}. 
	
	In this paper, our aim is to show that some of the results from the theory of monoid automata and group automata still hold for more general models and more general classes of monoids. In the first part, we will extend some results proven in \cite{MS01} for finite automata over commutative groups to rational monoid automata defined by finitely generated inverse permutable semigroups, which are a remarkable generalization of some commutative semigroups. We prove a lemma similar to the Interchange Lemma proven in \cite{MS01} which helps us to show that the language 
	$\mathtt{L_1}^*=  \{a^n b^n : n\geq 1\}^* $ can not be recognized by any rational monoid automaton over a finitely generated permutable monoid. It is also shown that the class of languages recognized by rational monoid automata over finitely generated completely simple or completely 0-simple permutable monoids is a semi-linear full trio. In the second part of the paper, we investigate valence pushdown automata and context-free valence grammars. Using the well known equivalence between pushdown automata and finite automata over polycyclic monoids, valence pushdown automata and finite valence automata turn out to be equivalent in terms of language recognition power. We also show that the results proven in \cite{Re10} for monoid automata can be easily verified for context-free valence grammars.

	\section{Background}
	
	 Let $ M $ be a monoid. We denote by $ \circ $ the binary operation of product of $ M $ and by 1 its identity. An \textit{(extended) finite automaton over $ M $} \cite{MS97} (also named an \textit{$ M $-automaton}) is a 6-tuple \[ \mathcal{E} = (Q, \Sigma,M,\delta, q_0,Q_a), \]
	where $Q$ is the set of
	states,  $\Sigma$ is the input alphabet, $q_0 \in Q$ denotes the initial state, $Q_a \subseteq Q$ denotes the
	set of accept states and the transition function $\delta$ is defined as
	\[\delta: Q \times \Sigma_{\varepsilon} \rightarrow \mathbb{P}(Q\times M),\] 
where $ \mathbb{P} $ denotes the power set and $ \Sigma_{\varepsilon}=\Sigma \cup \{\varepsilon\}  $ where $ \varepsilon $ denotes the empty string. For every $ (q \times \sigma) \in Q \times \Sigma_{\varepsilon} $, $ (q',m) \in \delta(q,\sigma)   $ means that when $\mathcal{E}$ reads $\sigma$
	in state $q$, it will move to state $q'$, and write $ x\circ m $ in the register, 
	where $ x $ is the old content of the register. The initial value of the register is the identity element of $ M $. The string is accepted if, after
	completely reading the string,  $\mathcal{E}$ enters an accept state with the 
	content of the register being equal to 1. 
	
	Extended finite automata over monoids are sometimes called \textit{monoid automata} as well. We will denote the family of languages accepted by $ M $-automata by $ \mathfrak{L}_1(M) $. In the cases where the monoid is a group, such models are called group automata.
	
Let $M$ be a monoid. An $M$-automaton is said to be \textit{with targets} if it is equipped with two subsets $ I_0, I_1 \subseteq M $ called the initial set and the terminal set respectively. An input string $ w \in \Sigma^* $ is accepted by the automaton if 
there exists a computation from the initial state to some accepting state such that $ x_0x \in I_1 $, where
$ x_0 \in I_0 $ and $ x  \in M$ is the content of the register of the machine after the reading of $w$. Recall that the \textit{rational subsets} of a monoid are the closure of its finite subsets under the rational operations union, product and Kleene star. In the case that  $ I_0$ and $ I_1 $ are rational subsets of $ M $, the model is called  \textit{rational monoid automaton} defined by $M$ \cite{Re10,RK10}.  The family of languages accepted by rational monoid automata defined by a monoid $M$ will be denoted by  $\mathfrak{L}_{Rat}(M)$. Note that the family of languages accepted by rational monoid automata 
where $I_0=I_1=\{1\}$ coincides with the set of languages recognized by ordinary $ M $-automata.

\textit{Rational semigroup automata} are defined analogously by taking $M$ as a semigroup instead of a monoid.
	
	Let $ G=(N,T,P,S) $ be a  \textit{context-free grammar}  where $ N$ is the nonterminal alphabet, $ T $ is the terminal alphabet, 
	$ P \subseteq N \times (N \cup T)^* $ is the set of rules or productions, and $ S \in N $ is the start symbol.  We will denote by $ \Rightarrow $ and $ \Rightarrow^* $ the step
	derivation relation and its regular closure respectively.
	$ \mathfrak{L}(G)$ denotes the language $\{w\in T^*: S \Rightarrow^* w\} $ of words generated by $G$. 
	
	Let us now recall the notion of valence context-free grammar introduced in  \cite{FS02}. 
	Given a monoid $M$, a \textit{context-free valence grammar over  $ M $} is a five-tuple  $ G=(N,T,P,S,M) $, where $ N,T,S  $ are defined as before and $ P \subseteq N \times (N \cup T)^* \times M $ is a finite set of objects called \textit{valence rules}. Every valence rule can be thus described as an ordered pair
	$ p=(A \rightarrow \alpha, m) $, where $(A \rightarrow \alpha)\in N \times (N \cup T)^*$ and  $ m\in M $. The element $m$ is called the valence of $ p $. 
	
	The step derivation $(\Rightarrow)$ of the valence grammar is defined as follows: if $ (w,m), (w',m')\in T^*\times M$, then  $ (w,m) \Rightarrow (w',m') $ if there exists a valence rule $ (A \rightarrow \alpha,n) $ such that $ w=w_1 Aw_2 $ and $ w'=w_1 \alpha w_2 $ and $ m'=mn$. 
	The regular closure of $ \Rightarrow$ will be denoted by $ \Rightarrow^*$.
	A derivation of $G$ will be said {\em successful} or {\em valid} if it is of the form $(S,1) \Rightarrow^* (w,1)$, that is, it
	transforms the pair $(S,1)$ into the pair $(w,1)$, after finitely many applications of the step derivation relation. 
	The language generated by $ G $ is the set $ \mathfrak{L}(G)$ of all the words $w$ of $T^*$ such that $(S,1) \Rightarrow^* (w,1)$. 
	
	A context-free valence grammar is said to be  \textit{regular} if all of its rules are right-linear, that is, every valence rule $ (A \rightarrow \alpha,n) $ is such that
	$\alpha = uX,$ where $u\in T^*$ and $X\in N$. 
	The language families generated by context-free and
	regular valence grammars over $ M $ are denoted by $ \mathfrak{L}(\textup{Val}, \textup{CF},M)$ and  $ \mathfrak{L}(\textup{Val}, \textup{REG},M)$, respectively.
	
	Let us finally recall the notion of valence pushdown automaton introduced in  \cite{FS02}. Let
	
	 $$  \mathcal{P}=(Q, \Sigma, \Gamma, \delta, q_0, Q_a) $$ be a finite nondeterministic pushdown automaton, where 
	$Q$ is the  set of states, $\Sigma$ is the input alphabet, $ \Gamma $ is the stack alphabet, $q_0$ is the initial state, $Q_a$ is the set of accept states, and 
	$$\delta: Q \times \Sigma_{\varepsilon} \times \Gamma_{\varepsilon} \rightarrow \mathbb{P}(Q \times \Gamma_{\varepsilon})$$
	is the transition function,
	where $\Gamma_{\varepsilon} = \Gamma \cup \{\varepsilon\}$. 
	%
	%
	Given a monoid $M$, a
	\textit{nondeterministic  valence pushdown automaton} \textup{(PDA)} over $ M $ is the model of computation obtained from $\cal P$ as follows:  
	with every transition of $\cal P$  is assigned an element of $M$, called {\em valence} of the transition. Then the valence of an arbitrary computation 
	is defined as the product of the valences of all the transitions of the computation (taken in the obvious order). 
	A word of $\Sigma ^*$ is said to be {\em accepted} by the model if there exists an accepting computation for the word whose valence is the identity of $M$.
	The set of all the accepted words  is defined as the language accepted by the valence pushdown automaton. 
	The family of languages accepted by  valence \textup{PDA} over $ M $ is denoted  $ \mathfrak{L}(\textup{Val},\textup{PDA},M)$.
	It is worth noticing that the equivalence between valence pushdown automata and valence context-free grammars does not hold for an arbitrary monoid. 
	However a remarkable result of \cite{FS02} shows that such equivalence is true 
	if $ M $ is a commutative monoid.
	
	Let us finally recall that in  the case  the pushdown automaton $  \mathcal{P}$ is a finite state automaton, the corresponding valence model has been called 
	\textit{valence automaton over $ M $}, and this model coincides with that of $ M $-automaton. 
	In particular, the family $ \mathfrak{L}(\textup{Val}, \textup{NFA},M)$ of languages accepted by valence automata  coincides with  $\mathfrak{L}_1(M)$, and 
	one can prove that $ \mathfrak{L}(\textup{Val}, \textup{NFA},M) = \mathfrak{L}(\textup{Val}, \textup{REG},M)$ \cite{FS02}. 
	
	Throughout the paper, we will denote by $ \mathsf{REG} $ and $ \mathsf{CF} $  the class of regular and context-free languages respectively. 

\section{Rational monoid automata over permutable monoids}
In this part of the paper, we will generalize some results proved in \cite{MS01}. 
This generalization is based upon two different concepts: the permutation property for semigroups and monoid automata with  rational targets.

\subsection{Permutation property for semigroups}\label{Sec3.1}
Let us talk about the notion of
permutation property. We assume that the reader is familiar with the algebraic theory of semigroups (see \cite{Ho95,La79}). The interested reader can find in \cite{LV99} an excellent survey on this topic.

Let $n$ be a positive integer and let ${\cal S}_n$ be the symmetric group of order $n$.

Let $S$ be a semigroup and let $n$ be an integer with $n\geq2$.  We say that $S$ is {\it n-permutable}, or that $S$ 
satisfies the property {\em ${\cal P}_n$}, if, for every sequence of $n$ elements
$s_1, \ldots, s_n$ of $S$, there exists a permutation $\sigma \in {\cal S}_n$, different from the identity, such that
$$s_1 s_2 \cdots s_n = s_ {\sigma(1)}s_{\sigma(2)}\cdots s_ {\sigma(n)}.$$

A semigroup $S$ is said to be {\em permutable} or that $S$ satisfies the {\em permutation property} if there exists some $n\geq 2$ such that $S$ is $ n $-permutable. 

Obviously the property ${\cal P}_2$ is equivalent to commutativity. If $S$ is a finite semigroup of cardinality $n$, then one immediately verifies that
$S$ is $r$-permutable, with $r\geq n+1$.  The permutation property  was introduced and studied by Restivo and 
Reutenauer in 1984 \cite{RR84} as a finiteness condition for
semigroups. 

Let us recall that an element $s$ of a semigroup $S$ is said to be {\em periodic} if there exist two integers $i, j$, with
$1\leq i < j$,  such that $s^i = s^j$. In particular, if $s=s^2$, $s$ is called {\em idempotent}. Moreover, if $S$ is a group and $1$ is its identity,
then $x$ is periodic if and only if $x^n=1$, for some positive integer $n$. 
If every element of $S$ is periodic, then $S$ is said to be periodic. 

The  following result holds.
\begin{fact}\textup{\cite{RR84}}\label{factRR}
	Let $S$ be a finitely generated and periodic semigroup. Then $S$ is finite if and only if $S$ is permutable. 
\end{fact}
In the case of finitely generated groups,  Curzio, Longobardi and Maj proved a remarkable algebraic characterization of permutable groups.

\begin{fact}\textup{\cite{CLM83}}\label{factCLM}
	Let $G$ be a finitely generated group. Then $G$  is permutable if and only if $G$ is Abelian-by-finite, {i.e.} G has  a (normal) Abelian subgroup of finite index. 
\end{fact}

One of the most important class of semigroups is that of {inverse monoids} (see \cite{Ho95,La79,Ok91}). A monoid $M$ is said to be  {\em inverse} if every element $m\in M$ possesses
a unique element $m'$, called the {\em inverse of $m$}, such that $m=mm'm$ and $m' = m'mm'$ (see \cite{Ho95,Ok91}). An important class of inverse monoids called the polycyclic monoids will be further explored in Section 4.


%

A remarkable result of Okni{\'n}ski
provides a characterization of the permutation property for  finitely generated  inverse monoids \cite{Ok91}.

\begin{fact}\textup{\cite{Ok91}}\label{factOK}
	Let $M$ be a finitely generated inverse monoid. Then the following conditions are equivalent:
	\begin{enumerate}[(i)]
		\item $M$  is permutable.
		\item The set of idempotent elements $E(M)$ of $M$ is finite and every subgroup of $M$ is finitely generated and Abelian-by-finite.
	\end{enumerate}
	
\end{fact}
The following result is a straightforward consequence of Fact \ref{factOK} and some basic facts of inverse monoids.
For this purpose, we recall that the algebraic structure of a monoid $M$ is described by its  Green's relations:
$\cal L, R, H, D,$ and $\cal J$. In particular, we recall that, for every $m, m'\in M$, $m \ {\cal R} \ m'$ (resp., $m\  {\cal L}\  m'$)
if $mM = m'M$ (resp., $Mm = Mm'$), and  ${\cal H} = {\cal R}\cap \cal L$
(the interested reader is referred to Ch. II and Ch. V of \cite{Ho95}, or Ch. III of \cite{LV99}). 
\begin{proposition}\label{cor-ok}
	Let $M$ be a finitely generated inverse monoid. Then $M$ is finite
	if and only if  $M$ is permutable and every finitely generated subgroup of $M$ is periodic. 
\end{proposition}

\begin{proof}
	If $M$ is finite then trivially it is periodic and, by Fact \ref{factRR} it is permutable. 
	
	Let us prove the converse. Suppose $ M $ is permutable and every finitely generated subgroup of $ M $ is periodic.

	As an immediate consequence of a well-known property of finitely generated semigroups
	(\cite{LV99}, Ch. III Prop. 3.2.4 or \cite{La79} Lem. 3.4), the ${\cal H}_e$-class of an arbitrary idempotent $e$ of $M$ is a 
	finitely generated maximal subgroup of $M$. Since every finitely generated subgroup of $ M $ is periodic, ${\cal H}_e$ is periodic. Since $M$ is inverse and permutable, by Fact \ref{factOK}, the set of idempotent elements $E(M)$ of $M$ is finite and every finitely
	generated subgroup of $M$ is Abelian-by-finite. Therefore, ${\cal H}_e$ is Abelian-by-finite. Now by Fact \ref{factCLM}, it follows that ${\cal H}_e$ is permutable. Since ${\cal H}_e$ is periodic and permutable, we conclude that ${\cal H}_e$ is finite by Fact \ref{factRR}.
	
	Moreover, always as a consequence of the fact that $M$ is inverse,  every $\cal R$-class  and every $\cal L$-class of $M$ contain exactly one idempotent
	(\cite{Ho95}, Ch. V Thm. 1.2).	Since $E(M)$ is finite, it follows that the number of $\cal R$-classes and $\cal L$-classes of $M$ is finite which implies that the number of $\cal H$-classes of $M$ is finite. Recall now that, for every $\cal R$-class $R$ of $M$,
	the cardinality of an arbitrary $\cal H$-class contained in $R$ equals the cardinality of the unique ${\cal H}_e$-class contained 
	in  $R$, whose representative is an idempotent $e$ (\cite{Ho95}, Ch. II Lem. 2.3). Since we have proved that, for every idempotent $ e $, the $ {\cal H}_e $-class is finite, by the latter, we then get that every $ \cal H $-class is finite which implies that  $M$ is finite. 
\end{proof}

\subsection{New results on rational monoid automata over permutable monoids}

Let $\cal M$  be the family of finitely generated inverse permutable monoids.
We now prove that some meaningful properties proved in \cite{MS01} in the case of  extended finite automata over finitely generated Abelian groups can be lifted to
rational monoid automata defined by monoids of $\cal M$. 

We start with the following basic facts.

\begin{lemma}\label{bas-fac}
	Let $M$ be a monoid of $\cal M$. Then one has:
	\begin{enumerate} [(i)]
		\item $\mathfrak{L}_{Rat}(M)$ contains the family $\mathsf{REG}$.
		\item If $M$ has an infinite subgroup, then $\mathfrak{L}_{Rat}(M)$ contains the language   $\mathtt{L_{1}}= \{a^nb^n: n\geq 1\}$, over the alphabet $\Sigma = \{a, b\}$.
		
	\end{enumerate}
\end{lemma} 
\begin{proof}	\begin{enumerate} [(i)]
		\item Let $\mathtt{L}$ be a regular language accepted by a finite automaton ${\cal A} = (Q, \Sigma,M,\delta, q_0,Q_a)$.
		We can turn $\cal A$ into a rational monoid
		automaton  $\cal B$ by setting $I_0 =  I_1 =\{1\}$ and, 
		for every state $q\in Q$ and for every letter $a\in \Sigma_{\varepsilon}$,
		by replacing the state $q'$ such that $q'\in \delta (q, a)$ into the ordered pair $(q', 1)$.
		Then one easily checks that $\mathtt{L}$ is accepted by $\cal B$.
		
		\item Let $G$ be an infinite subgroup of $M$. Since $ M $ is inverse and permutable, by Fact \ref{factOK}, $G$ is finitely generated and Abelian-by-finite. By Fact \ref{factCLM}, it follows that $ G $ is permutable.
		Since $G$ is infinite, by Fact \ref{factRR}, $G$ has an element $x$ which is not periodic.
		Let $e$ be the identity of $G$. 
		Let $x'$ be the inverse of $x$ in $G$. 
		Let $\cal A$ be the rational monoid automaton over two states ${q_0, q_1}$, where $q_0$ is the (unique) initial state, $q_1$ is the (unique) final state,  $I_0=I_1=\{e\}$ and 
		the transition function $\delta$ of $\cal A$ is defined as: 
		$\delta (q_0,a) = (q_0, x),  \delta (q_0,b) = (q_1, x{'}),  \delta (q_1,b) = (q_1, x{'}),  \delta (q_1,a) = \emptyset.$
		Then one easily checks that $\mathtt{L_1}$ is accepted by $\cal A$. 
	\end{enumerate}
\end{proof}

In Thm. 1 of \cite{MS01}, it is proven that for any group $ G $, $ \mathfrak{L}(G)=\mathsf{REG} $ if and only if all finitely generated subgroups of $ G $ are finite. The following proposition gives a similar characterization for rational monoid automata defined over finitely generated inverse permutable monoids.

\begin{proposition}\label{thm1-mit-stib-equiv}
	Let $M$ be a monoid of $\cal M$. Then $\mathfrak{L}_{Rat}(M) = \mathsf{REG}$ if and only if every finitely generated subgroup of $M$ is periodic.	
	
\end{proposition} 

\begin{proof}
	Let us prove the necessity. By contradiction, assume the contrary. Hence there exists a subgroup of $M$ with an element which is not periodic. 
	By {Lemma} \ref{bas-fac}, one gets $ \{a^nb^n: n\geq 1\} \in \mathfrak{L}_{Rat}(M)$, so    $\mathfrak{L}_{Rat}(M) \neq \mathsf{REG}$.
	
	Let us prove the sufficiency.  By Proposition \ref{cor-ok},  $M$ is finite. If a language belongs to the set  $\mathfrak{L}_{Rat}(M)$, then it is accepted by a rational monoid automaton with initial set $ I_0=\{1\} $ \cite{Re10}.  
	Let us show that an arbitrary  rational monoid automaton 
	${\cal A} =(Q, \Sigma,M,\delta, q_0,Q_a)$ with rational targets $I_0 =\{1\}$ and $I_1$,
	can be simulated by a finite automaton. Indeed, let ${\cal B} =(Q\times M, \Sigma, \hat{\delta}, q_0\times I_0, Q_a\times I_1)$ be the nondeterministic finite automaton (with $\varepsilon$-moves), where 
	the transition function  $\hat{\delta}$ of $\cal B$ is defined as
	$$\hat{\delta} (\langle q,m \rangle , a):= \{\langle q', m'\rangle \in Q\times M : (q', x)\in \delta (q, a) \ \mbox{with} \  m'= mx \}$$
	for every $\langle q, m \rangle \in Q\times M$, and for
	every $a\in \Sigma_{\varepsilon}$.  
	Since $M$ is finite, $\cal B$ is well defined and it can be easily checked, by induction on the length of the computation that
	spells a word $u$, that $u$ is accepted by $\cal A$ if and only if $u$ is accepted by $\cal B$. This proves $\mathfrak{L}_{Rat}(M)\subseteq \mathsf{REG}$.
	Since by Lemma \ref{bas-fac}, $\mathsf{REG} \subseteq \mathfrak{L}_{Rat}(M)$, from the latter, we get $\mathsf{REG} = \mathfrak{L}_{Rat}(M)$.
\end{proof}

We now prove that the language $\mathtt{L_1}^* =  \{a^n b^n : n\geq 1\}^*$ is not in the family of $\mathfrak{L}_{Rat} ({\cal M})$. In order to achieve this result, we prove a lemma that is similar to the  ``Interchange Lemma" proven in \cite{MS01} Lem. 2 for Abelian groups.
\begin{lemma}\label{il-fla}
	Let $M$ be a $k$-permutable monoid and let $\mathtt{L} \in \mathfrak{L}_{Rat}(M)$. Then there exists a positive integer $m$ such that, for every word $w\in \mathtt{L}$, with
	$|w|\geq m$, and, for every factorization of 
	$w = w_1 w_2 \cdots w_m,  |w_i| \geq 1 \ (1\leq i \leq m),$
	there exist integers  $0\leq i_0 < i_1 < i_2 < \cdots < i_k  < i_{k+1} \leq m$ such that 
	
	\begin{equation}\label{il-fla0}
	w = \lambda W_1 W_2 \cdots W_k \mu,
	\end{equation}
	where
	\begin{equation}\label{il-fla00}
	\lambda = w_1 \cdots w_{i_0}, \quad \mu = w_{i_{k+1}} \cdots w_ m,\footnote{it is understood that if $i_0=0$ 
		(resp., $i_{k+1} = m$), then  $\lambda = \varepsilon$ (resp., $\mu = \varepsilon$).}
	\end{equation}
	and, for every $j=1, \ldots, k$,
	\begin{equation}\label{il-fla000}
	W_j= w_{1+ i_{j-1}} \cdots w_{i_{j}}
	\end{equation}
	\item there is a permutation $\sigma \in {\cal S}_k$, different from the identity, such that the word
	$w_\sigma = \lambda W_{\sigma(1)} W_{\sigma(2)}  \cdots W _{\sigma(k)} \mu$
	is in $\mathtt{L}$. 
	
\end{lemma}

\begin{proof}
	By hypothesis, there exists a rational monoid automaton $\cal A$, defined by a $k$-permutable monoid $M$, that
	accepts $\mathtt{L}$. Let ${\cal A} =(Q, \Sigma,M,\delta, q_0,Q_a)$ with rational targets $I_0$ and $I_1$. If $c$ is an arbitrary computation of $\cal A$,
	the element of $M$ associated with $c$ will be denoted by $m(c)$.

	Let $w$ be a word of $\mathtt{L}$ and
	let $c$ be a successful computation of $\cal A$ such that $c$ spells $w$. In particular,  
	there exists some $i_0\in I_0$ where $i_0m(c)\in I_1$. 
	
	Let $m = \max \{k, n\}^2 +1$ where $n$ is the number of states of $\cal A$. 
	Suppose now  that $|w|\geq m^2$.  By using the pigeonhole principle, there exists a state $q$ of $\cal A$
	such that $c$ can be factorized as the product of computations
	$$c = c_\lambda c_1 c_2 \cdots c_k c_\mu,$$
	where\begin{itemize}
		\item $c_\lambda = q_0 \rightarrow q,$ is a computation that spells $\lambda$;
		\item $c_\mu = q \rightarrow q_f,$ with $q_f\in Q_a$, is a computation that spells $\mu$;
		\item for every $i=1, \ldots, k,$ $c_i=q \rightarrow q$, is a computation that spells $W_i$, ($1\leq i \leq k$),
	\end{itemize}
	and $\lambda, \mu,$ and $W_i$, with $1\leq i \leq k$, are defined as in (\ref{il-fla00}) and (\ref{il-fla000}) respectively. 
	Hence, from the latter, and taking into account that $M$ is $k$-permutable, we have that there exists a
	permutation $\sigma \in {\cal S} _k \setminus \{id.\}$, such that
	\begin{equation}\label{il-fla2}
	m (c) = m(c_\lambda) m(c_1)  \cdots m(c_k) m(c_\mu)=  m(c_\lambda) m(c_{\sigma(1)})   \cdots m(c_{\sigma(k)}) m(c_\mu).
	\end{equation}
	Let $c_\sigma$ be the product of computations $c_\sigma = c_\lambda c_{\sigma(1)}   \cdots c_{\sigma(k)} c_\mu.$
	Observe that, by the definition of the computation $c_i, 1\leq i\leq k$, $c_\sigma$ is well defined as a computation of $\cal A$,
	$c_\sigma= s_0 \rightarrow s_f,$ from $s_0$ to $s_f$.  
	Moreover, by (\ref{il-fla2}), one has $m(c_\sigma) = m(c)$. Since  $c_\sigma$ spells $w_\sigma$, one has  $w_\sigma \in L$. This completes the proof. 
\end{proof}

It is shown that $\mathtt{L_1}^*=  \{a^n b^n : n\geq 1\}^* $ can not be recognized by any finite automaton over an Abelian group in \cite{MS01} Prop. 2. Now we prove that the same language is not in $  \mathfrak{L}_{Rat}(M) $ when $ M $ is a finitely generated permutable monoid.

\begin{corollary}\label{il-fla-cor}
	Let $M$ be a finitely generated permutable monoid. Then $\mathtt{L_1}^*=  \{a^n b^n : n\geq 1\}^* \notin \mathfrak{L}_{Rat}(M)$. 
\end{corollary}

\begin{proof}
	By contradiction, assume the contrary. Thus there exists a rational monoid automaton $\cal A$ defined by a permutable monoid
	such that $\mathtt{L_1}^*$ is accepted by $\cal A$. Let us consider the word of $\mathtt{L_1}^*$, 
	$w = (ab)(a^2b^2) \cdots (a^\ell b^\ell)$, for $\ell\in \N$ sufficiently large. Let us consider the following factorization for $w$:
	\begin{equation}\label{eqfla}
	a\cdot ba^2 \cdot b^2 a^3  \cdots b^{i-1}a^i \cdot b^i a^{i+1} \cdots a^{\ell-1}\cdot  b^{\ell -1} a^\ell \cdot b^\ell,
	\end{equation}
	that is, $w= w_1 \cdots w_{\ell +1},$ where
	$w_i = b^{i-1} a^ {i}, \ i=1, \ldots, \ell,$ and $ w_{\ell+1} = b^\ell.$
	Let us apply Lemma \ref{il-fla} to $w$ with respect to the  factorization (\ref{eqfla}).
	Hence we can write $w = \lambda W_1 W_2 \cdots W_k \mu$ in the form of the factorization (\ref{il-fla0}) of {Lemma} \ref{il-fla},
	and  there exists a permutation $\sigma \in {\cal S}_k\setminus \{id.\}$, such that the word
	$w_\sigma = \lambda W_{\sigma(1)} W_{\sigma(2)}  \cdots W _{\sigma(k)} \mu$ is in $\mathtt{L_1}^*$. 
	
	Let $\underline{k}= \{1, \ldots, k\}$ and let $i$ be the minimal number of $\underline{k}$ such that $\sigma (i) \neq i$. 
	Since $\sigma \neq id.$, this number exists  and $i<k$. Moreover the number  $j\in  \underline{k}$ such that $i = \sigma (j)$ is strictly
	larger than $i$. Indeed, otherwise $j\leq i$, would  imply $j < i,$ so contradicting the minimality of $i$. 
	Now, according to (\ref{il-fla00}) and (\ref{il-fla000}) of {Lemma} \ref{il-fla}, and taking into account that $i<j$, there exist positive integers $\alpha, \beta$,
	with $ \beta > 1  + \alpha$ such that
	
	$$\lambda W_1 \cdots W_{i-1} = a(ba^2) \cdots (b^{\alpha}a^{1+ \alpha}),$$
	where 
	$$W_{j} =  (b^{\beta}a^{1+ \beta}) \cdots (b^{ \gamma}a^{1+ \gamma}), \quad \beta \leq \gamma.$$
	
	Hence we get 	
	\begin{align*}	
	w_\sigma &= \lambda W_{1} \cdots W_{i-1} W_j W _{\sigma(i+1)}  \cdots W _{\sigma(k)} \mu\\
	&=(aba^2 \cdots b^{ \alpha}a^{1+ \alpha})(b^{\beta}a^{1 + \beta} \cdots b^{ \gamma}a^{1 + \gamma})u, \quad u\in A^*,
	\end{align*}
	so that $aba^2 \cdots b^{\alpha}a^{1+ \alpha}b^{ \beta}$ is a prefix of $w_\sigma$.  Since $\beta > 1 + \alpha$,
	the latter contradicts the fact that $w_\sigma \in \mathtt{L_1}^*.$ 
	Hence $\mathtt{L_1}^*\notin \mathfrak{L}_{Rat}(M)$.  \end{proof}

We can conclude that the following closure properties proven in \cite{MS01} Thm. 8 for finite automata over commutative groups also hold for rational monoid automata over finitely generated inverse permutable monoids.

\begin{corollary}\label{cor-mitstib}
	Let $M$ be a monoid of $\cal M$. Then either $\mathfrak{L}_{Rat}(M) = \mathsf{REG}$, or $\mathfrak{L}_{Rat}(M)$ is closed neither under Kleene star $^*$, nor
	under substitutions.
\end{corollary}

\begin{proof}
	If every finitely generated subgroup of $M$ is periodic, by Proposition \ref{thm1-mit-stib-equiv}, $\mathfrak{L}_{Rat}(M) = \mathsf{REG}$. 
	Otherwise, by Lemma \ref{bas-fac}, $\mathfrak{L}_{Rat}(M)$ contains the language   $\mathtt{L_1}= \{a^nb^n: n\geq 1\}$ and, by Corollary \ref{il-fla-cor},
	$\mathtt{L_1} ^* \notin \mathfrak{L}_{Rat}(M)$. 
	The non-closure under substitutions follows since $\mathfrak{L}_{Rat}(M)$ contains $\mathsf{REG}$ and it is not closed under Kleene star. 
\end{proof}

We finally discuss a property of the class of languages recognized by rational monoid automata. First, let us recall some definitions about semigroups.

Let $ S $ be a semigroup. $ S^1 $ is the semigroup obtained from $ S $ by adjoining an identity element to $ S $. Similarly $ S^0 $ is the semigroup obtained from $ S $ by adjoining a zero element to $ S $.

An \textit{ideal} $ I $ of a semigroup $ S $ is a subset of $ S $ with the property that $ S^1IS^1 \subseteq I $. A semigroup is called \textit{simple} if it contains no proper ideal. A semigroup $ S $  with a zero element is called 0-$ simple $ if the only ideals of $ S $ are $ \{0\} $ and $ S $ itself, and $ SS \neq \{0\} $.

An idempotent element is called \textit{primitive} if for every non-zero idempotent $ f $, $ ef=fe=f $ implies that $ e=f $. A semigroup is \textit{completely simple} (\textit{completely 0-simple}) if it is simple (0-simple) and contains a primitive idempotent.

Among the closure properties fulfilled by $\mathfrak{L}_{Rat}(M)$, it is known that, for an arbitrary finitely generated monoid $M$, $\mathfrak{L}_{Rat}(M)$ is a full trio \cite{Re10}.
We recall that a family of languages is said to be a full trio  if
it is closed under taking morphisms, intersection with regular languages, and inverse morphisms.  
This can be reinforced in the case of monoids that are completely simple or completely $0$-simple (see \cite{Ho95}, Ch. III for a basic introduction 
to such structures).

For this purpose, let us recall the following theorem of \cite{RK10}. 
\begin{fact}\textup{\cite{RK10}}\label{factRK} 	Let $M$ be a completely simple or completely $0$-simple monoid with maximal non-zero subgroup $G$. 
	Then $\mathfrak{L}_{Rat}(M) = \mathfrak{L}_{Rat}(G)= \mathfrak{L}_{1}(G)$. 
\end{fact}

A subset $ S \subseteq \mathbb{N}^n  $ is a \textit{linear} set if $ S=\{v_0 + \Sigma_{i=1}^k c_iv_i | c_1,\dots,c_k \in \mathbb{N}\} $ for some $ v_0,\dots,v_k \in\mathbb{N}^n  $. A \textit{semi-linear} set is a finite union of linear sets. A full trio is called \textit{semi-linear} if the 
Parikh image of every language of the family is semi-linear. 
Moreover a language $\mathtt{L}\subseteq A^*$ is said to be {\em bounded} if there exist words $u_1, \ldots, u_n \in A^+$ such that
$\mathtt{L} \subseteq u_1 ^* \cdots u_n^*$. A bounded language is said to be {\em (bounded) semi-linear} if there exists a semi-linear set $\cal B$
of $\N^n$ such that $\mathtt{L} = \{u_1 ^{b_1} \cdots u_n ^{b_n} : (b_1, \ldots, b_n)\in {\cal B}\}$. 
We have then the following corollary. 
\begin{corollary}\label{flafin}
	Let $M$ be a finitely generated completely simple or completely $0$-simple monoid with maximal non-zero subgroup $G$. 
	If $M$ is permutable, then $\mathfrak{L}_{Rat}(M)$  is a semi-linear full trio. In particular, every bounded language in $\mathfrak{L}_{Rat}(M)$  is (bounded) semi-linear.  
\end{corollary}
\begin{proof}
	Let $G$ be the maximal subgroup of $M$. By Fact \ref{factRK}, $ \mathfrak{L}_{Rat}(M)=\mathfrak{L}_{1}(G) $. Since  $M$ is a finitely generated   semigroup, then $G$ is a finitely generated
	group as well (\cite{LV99}, Ch. III Prop. 3.2.4 or \cite{La79} Lem. 3.4). Moreover, since $G$ is permutable,  
	by Fact \ref{factCLM}, $G$ has a finitely generated Abelian subgroup $H$ of finite index in $G$. By Cor. 3.3 of \cite{Co05}, $ \mathfrak{L}_{1}(G)=\mathfrak{L}_{1}(H) $ and by Thm. 7 of \cite{MS01}, one has $\mathfrak{L}_{1}(H) = \mathfrak{L}_{1}(\Z^m)$ for some $ m\geq 1 $, where $\Z^m$ is the
	free Abelian group over $m$ generators. 
	We now recall that, by a result of \cite{Gr78} $\mathfrak{L}_{1}(\Z^m)$ is a semi-linear full trio.
	Let $\mathtt{L}$ be a bounded language in $\mathfrak{L}_{Rat}(M)$. By the previous considerations, $\mathtt{L}\in \mathfrak{L}_{1}(\Z^m)$ and the fact that
	$\mathtt{L}$  is (bounded) semi-linear follows from a classical theorem of Ibarra
	\cite{IS15}.
\end{proof}

\section{Valence grammars and valence automata}
In this section, we are going to focus on valence grammars and valence automata. We start by making a connection between valence automata and valence PDA, proving that valence PDA are only as powerful as valence automata. Then, we extend some results proven for regular valence grammars to context-free valence grammars.

\subsection{Equivalence of finite and pushdown automata with valences}\label{sec: pdav}

A well known class of inverse monoids is that of the polycyclic monoids. Let $ X $ be a finite alphabet and let $X^*$ be the free monoid of words over $X$. For each symbol $ x\in X $, let $P_x$ and $Q_x$ be functions from $X^*$ into $X^*$
defined as follows:  for every $u\in X^*$,
\begin{align*}
P_x(u)= ux,\quad 
Q_x(ux)=u.
\end{align*}
Note that $Q_x$ is a partial function from $X^*$ into $X^*$ whose domain is the language $X^*x$. 
The submonoid of the monoid of all partial functions on $X^*$ generated by the
set of functions $\{P_x,\, Q_x \ |\ \ x \in X   \}$ turns out to be an inverse monoid, denoted by $P(X)$, called the \textit{polycyclic monoid} on $X$. It was explicitly studied by Nivat and Perrot in \cite{NP70} and, in the case $ |X|=1 $, the monoid $ P(X) $ coincides with
the well-known structure of \textit{bicyclic monoid} which will be denoted by B (see \cite{Ho95}). 
Polycyclic monoids have several applications in formal language theory and, in particular, 
define an interesting storage model of computation for the recognization of formal languages \cite{Co05,Gi96,Ka09,Ni70,NP70}. 
%
%

For any element $ x \in X, $ $P_xQ_x=1 $ where $ 1 $ is the identity element of $ P(X) $ and for any two distinct elements $ x, y \in X $, $ P_xQ_y $ is the empty partial function which represents the zero element of $ P(X) $. The partial functions $ \{P_x,Q_x\} $ model the operation of pushing and popping $ x$ in a PDA, respectively. In order to model popping and pushing the empty string, let us define $ P_{\varepsilon}$ and $Q_{\varepsilon}$ as $ P_{\varepsilon}=Q_{\varepsilon}=1 $. The equivalence between PDA with stack alphabet X and P(X)-automata is investigated in various papers \cite{Co05,Gi96,Ka09}. Note that a $ P(X) $-automaton is a valence automaton over $ P(X) $. 

We will focus on the polycyclic monoid of rank 2, which will be denoted by $ \textup{P}_2 $, since it contains every polycyclic monoid of countable rank.

\begin{theorem}\label{thm: pdap2} For any monoid $ M $, 
	$\mathfrak{L}(\textup{Val},\textup{PDA},M) = \mathfrak{L}(\textup{Val},\textup{NFA},\textup{ P}_2 \times M ) $.
\end{theorem}
\begin{proof}Let $ \mathtt{L}\in \mathfrak{L}(\textup{Val},\textup{PDA},M)$ and $\mathcal{P}=\{Q,\Sigma,X,\delta,q_0, F,M\} $ be a valence PDA recognizing $ \mathtt{L}$. We know that a PDA with stack alphabet $ X $ is equivalent to a valence automaton over $ P(X) $. Hence, $ \mathcal{P} $ can be seen as an \textup{NFA} where two distinct valences (one in $P(X) $ and one in $ M $) are assigned to each transition. An equivalent valence automaton $ \mathcal{M}= \{Q,\Sigma,P(X) \times M,\delta',q_0, F\} $ can be constructed, where a valence from the monoid $P(X) \times M  $ is assigned to each transition. Recall that the partial functions $ Q_a $ and $ P_b $ model the operations of popping $ a $ and pushing $ b $ respectively. A transition of $ \mathcal{P} $ of the form $ (q',b,m) \in \delta(q,\sigma,a) $ where $ a,b \in X_{\varepsilon} $, $ q,q' \in Q $, $ \sigma \in \Sigma_{\varepsilon} $ and $ m\in M $ can be expressed equivalently as $(q', \langle Q_aP_b , m \rangle ) \in \delta'(q,\sigma) $ where $ \langle Q_aP_b , m \rangle  \in P(X) \times M  $. 

	A string is accepted by $ \mathcal{M} $ if and only if the product of the valences labeling the transitions in $ \mathcal{M} $ is equal to $ \langle 1,1 \rangle  $, equivalently when the product of the valences labeling the transitions in $ \mathcal{P} $ is equal to the identity element of $ M $ and the stack is empty. Since any polycyclic monoid is embedded in $ \textup{P}_2 $, we conclude that $ \mathtt{L}\in \mathfrak{L}(\textup{Val},\textup{NFA},\textup{P}_2\times M ) $.
	
	Conversely, let $ \mathtt{L}\in \mathfrak{L}(\textup{Val},\textup{NFA},\textup{P}_2\times M ) $ and let $ \mathcal{M} = \{Q,\Sigma,\textup{P}_2\times M,\delta,q_0, F\}$ be a valence automaton over $ \textup{P}_2\times M $ recognizing $ \mathtt{L} $. Suppose that $ \langle p, m \rangle \in \textup{P}_2\times M  $ is a valence labeling a transition of $ \mathcal{M} $. The product of the labels of a computation which involves a transition labeled by the zero element of $ \textup{P}_2 $ can not be equal to the identity element. Hence we can remove such transitions. Any nonzero element $ p $ of $ \textup{P}_2 $ can be written as $ Q_{x_{1}}Q_{x_{2}}\dots Q_{x_{n}} P_{y_{1}}P_{y_{2}}\dots P_{y_{o}} $ for some $ n,o \in \mathbb{N} $ and $ x_{i}, y_{i} \in X_{\varepsilon} $, after canceling out elements of the form $ P_aQ_a $ and $ P_bQ_b $, where $ X=\{a,b\} $ is the generator set for $ \textup{P}_2 $. The product can be interpreted as a series of pop operations followed by a series of push operations performed by a PDA, without consuming any input symbol. Hence, an equivalent valence PDA $ \mathcal{P}=\{Q',\Sigma,X,\delta',q_0, F,M\} $ can be constructed where a valence from $ M $ is assigned to each transition. Let $ (q', \langle p, m \rangle) \in \delta(q,\sigma) $ where  $ q,q' \in Q $, $ \sigma \in \Sigma_{\varepsilon} $, $\langle p, m \rangle \in \textup{P}_2\times M  $ and $ p=Q_{x_{1}}Q_{x_{2}}\dots Q_{x_{n}} P_{y_{1}}P_{y_{2}}\dots P_{y_{o}} $  be a transition in $ \mathcal{M} $. In $ \mathcal{P} $, we need an extra $ n+o $ states $ \{q_1, \dots ,q_{n+o}\} \notin Q $ and the following transitions to mimic that specific transition of $ \mathcal{M} $. 
	\begin{align*}
	(q_{1}, \varepsilon,m) &\in \delta'(q,\sigma,x_{1})  \\
	(q_{2}, \varepsilon,1) &\in \delta'(q_{1},\varepsilon,x_{2})\\
	&~\vdots\\
	(q_{n+1}, \varepsilon,1) &\in \delta'(q_{n},\varepsilon,x_{n})\\
     \end{align*}
     \begin{align*}
	(q_{n+2}, y_{1},1) &\in \delta'(q_{n+1},\varepsilon,\varepsilon )\\
	(q_{n+3}, y_{2},1) &\in \delta'(q_{n+2},\varepsilon,\varepsilon )\\
	&~\vdots\\
	(q', y_{o},1) &\in \delta'(q_{n+o},\varepsilon,\varepsilon )\\
	\end{align*}
	
	A string is accepted by $ \mathcal{P} $ if and only if the product of the valences labeling the transitions in $ \mathcal{P} $ is equal to the identity element of $M $ and the stack is empty, equivalently when the product of the valences labeling the transitions in $ \mathcal{M} $ is equal to $ \langle 1,1 \rangle $. We conclude that $ \mathtt{L}\in \mathfrak{L}(\textup{Val},\textup{PDA},M)$.
\end{proof}

Note that when $ M $ is commutative, the equality $\mathfrak{L}(\textup{Val},\textup{CF},M) = \mathfrak{L}(\textup{Val},\textup{NFA},\textup{P}_2\times M ) $ also holds. 
\begin{corollary}
	Let $ M $ be a polycyclic monoid of rank 2 or more. Then  $\mathfrak{L}(\textup{Val},\textup{PDA},M) $ is the class of recursively enumerable languages.
\end{corollary}
\begin{proof}
	It is known that $ \mathfrak{L}(\textup{Val},\textup{NFA}, M \times M  ) $ is the class of recursively enumerable languages \cite{Ka09} when $ M $ is a polycyclic monoid of rank 2 or more. Since $\mathfrak{L}(\textup{Val},\textup{PDA},M) = \mathfrak{L}(\textup{Val},\textup{NFA},\textup{P}_2\times M ) $, by Theorem \ref{thm: pdap2}, the result follows. 
\end{proof}

\subsection{Context-free valence languages}

It is known that the class of languages generated by regular valence grammars and the class of languages recognized by valence automata coincide \cite{FS02}. In this section, we are going to prove that the results proven in $ \cite{Re10} $ which hold for valence automata and therefore regular valence grammars, also hold for context-free valence grammars. Although the proofs are almost identical, they are presented here for completeness. Note that the same proofs can be also adapted to valence PDA.

Let $ I $ be an ideal of a semigroup $ S $. The binary relation $ \rho_I $ defined by 
\[ 
a \rho_I b \iff \mbox{ either $ a=b $ or both $ a $ and $ b $ belong to } I
\]
is a congruence. The equivalence classes of $ S\mod \rho_I$ are $ I $ itself and every one-element set $ \{x\} $ with $ x \in S\setminus I $. The quotient semigroup $ S/\rho_I $ is written as $ S / I $ and is called the \textit{Rees quotient semigroup}   \cite{Ho95}. 
\[ 
S /I = \{I\} \cup \{\{x\}| x \in S \setminus I \}
\]

In \cite{Re10} Prop. 4.1.1, it is shown that the elements belonging to a proper ideal of a monoid do not have any use in the corresponding monoid automaton. We show that the same result holds for context-free valence grammars.

\begin{proposition}\label{prop: ideal}
	Let $ I $ be a proper ideal of a monoid $ M $. Then $ \mathfrak{L}(\textup{Val},\textup{CF},M) = \mathfrak{L}(\textup{Val},\textup{CF},M/I).$
\end{proposition}
\begin{proof}	Let $ \mathtt{L} \in \mathfrak{L}(\textup{Val},\textup{CF},M)  $ and let $ G $ be a context-free grammar over the monoid $ M $ such that $ \mathfrak{L}(G)=\mathtt{L} $. The product of the valences which appear in a derivation containing a rule with valence $ x \in I  $, will itself belong to $ I $. Since $ I $ is a proper ideal and $ 1\notin I $, such a derivation is not valid. Hence any such rules can be removed from the grammar and we can assume that $ G $ has no such rules. For any $ x_1,x_2,\dots,x_n \in M \setminus I $, it follows that $ x_1\dots x_n=1 $ in $ M $  if and only if $ \{x_1\}\{x_2\}\dots \{x_n\}=\{1\} $ in $ M/I $. Let $ G' $ be the context-free grammar with valences in $ M/I $, obtained from $ G $ by replacing each valence $ x \in M$ with $ \{x\} $. It follows that a string $ w $ has a valid derivation in $ G $ if and only if the product of the valences is mapped to $ \{1\} $ in $ G' $. Hence  $ \mathfrak{L}(G')=\mathtt{L} $.
	
	Conversely let $ \mathtt{L} \in \mathfrak{L}(\textup{Val},\textup{CF},M/I)  $ and let $ G' $ be a context-free grammar over the monoid $ M/I $ such that $ \mathfrak{L}(G')=\mathtt{L} $. Suppose that there exists a valid derivation consisting of a rule with $ I $ as the valence. Then the product of the valences of the whole derivation will be $ I $, which is not possible. Let $ G $ be the context-free grammar with valences in $ M$, obtained from $ G' $ by replacing each valence $ \{x\} \in M/I$ with $ x $. Since  $ \{x_1\}\{x_2\}\dots \{x_n\}=\{1\} $ in $ M/I $ if and only if  $ x_1\dots x_n=1 $ in $ M $, a string $ w $ has a valid derivation in $ G $ if and only if the product of the valences is mapped to $ \{1\} $ in $ G' $. Hence  $ \mathfrak{L}(G)=\mathtt{L} $. 
\end{proof}

Let $ S $ be a semigroup. $ S $ is the \textit{null semigroup} if it has an absorbing element zero and if the product of any two elements in $ S $ is equal to zero. A null semigroup with two elements is denoted by $ \textup{O}_2 $.

The following corollary is analogous to  \cite{Re10} Cor. 4.1.2.
\begin{corollary}\label{cor: simple} 
	For every monoid $ M $, there is a simple or 0-simple monoid $ N $ such
	that \\  $ \mathfrak{L}(\textup{Val},\textup{CF},M) =   \mathfrak{L}(\textup{Val}, \textup{CF},N) $.
\end{corollary}
\begin{proof} 
	If $ M $ has no proper ideals then it is simple. Otherwise, let $ I $ be the union of all proper ideals of $ M $ and let $ N=M/ I $. We can conclude from the proof of Cor. 4.1.2 \cite{Re10} 
	that $ N^2 = {0} $ or $ N $ is 0-simple. If $ N^2 = {0} $, then $  N $ is  $ \textup{O}_2 $ and the semigroup  $ \textup{O}_2 $ does not add any power to the grammar since it does not even contain the identity element. Hence, $ \mathfrak{L}(\textup{Val},\textup{CF},  \textup{O}_2) = \mathfrak{L}(\textup{Val},\textup{CF},\{1\}) $ where $ \{1\} $ is the trivial monoid which is simple. In the latter case $ N $ is 0-simple and by Proposition \ref{prop: ideal}, $ \mathfrak{L}(\textup{Val},\textup{CF},M) = \mathfrak{L}(\textup{Val},\textup{CF},M/I) =  \mathfrak{L}(\textup{Val}, \textup{CF},N) $.
\end{proof}

Prop. 4.1.3 of \cite{Re10} states that a finite automaton over a monoid with a zero element is no more powerful then a finite automaton over a version of the same monoid from which the zero element has been removed, in terms of language recognition. The result is still true for context-free valence grammars since the same proof idea applies. The following notation is used: $ M^0 = M \cup \{0\} $ if $ M $ has no zero element and $ M^0=M $ otherwise.  
\begin{proposition}\label{prop: zero}
	Let $ M $ be a monoid. Then $ \mathfrak{L}(\textup{Val},\textup{CF},M^0) =\mathfrak{L}(\textup{Val},\textup{CF},M) $.
\end{proposition}
\begin{proof}
	Since $ M \subseteq M^0 $, it follows that $ \mathfrak{L}(\textup{Val},\textup{CF},M) \subseteq \mathfrak{L}(\textup{Val},\textup{CF},M^0) $. Suppose $ \mathtt{L} \in \mathfrak{L}(\textup{Val},\textup{CF},M^0) $ and let $ G $ be a context-free grammar with valences in $ M^0 $ and $ \mathfrak{L}(G)=\mathtt{L} $. Note that a valid derivation can not contain a rule with a zero valence since otherwise the product of the valences would be equal to zero. Any such rules can be removed from $ G $ to obtain $ G' $, a context-free grammar with valences in $ M $, without changing the language, and $ \mathtt{L} \in \mathfrak{L}(G') $.
\end{proof}

\begin{fact}\textup{\cite{Re10}}\label{fact: simple}
	A simple (0-simple) monoid with identity $ 1 $ is either a group (respectively, a group with 0 adjoined) or contains a copy of the bicyclic monoid as a
	submonoid having 1 as its identity element.
\end{fact}

Now we are ready to prove the main theorem of the section which will allow us to determine the properties of the set of languages generated by context-free valence grammars. We need the following proposition which is the grammar analogue of Prop. 1 of \cite{Ka06}.

\begin{proposition}\label{prop: kambites}
	Let $ M $ be a monoid, and suppose that $ \mathtt{L} $ is accepted by a context-free valence grammar over $ M $. Then there exists a finitely generated submonoid $  N $ of $ M $ such that $ \mathtt{L} $ is accepted by a context-free valence grammar over $ N $.
\end{proposition}
\begin{proof}
	There are only finitely many valences appearing in the rules of a grammar since the set of rules of a grammar is finite. Hence, the valences appearing in derivations are from the submonoid $ N $ of $ M $ generated by those elements. So the grammar can be viewed as a context-free valence grammar over $ N $.
\end{proof}

Recall that a group $ G $ is \textit{locally finite} if every finitely generated subgroup of $ G $ is finite.

\begin{theorem}\label{thm: characterize}
	Let $ M $ be a monoid. Then $ \mathfrak{L}(\textup{Val},\textup{CF},M) $ either
	\begin{enumerate}[(i)]
		\item equals $  \mathsf{CF} $,
		\item contains  $ \mathfrak{L}(\textup{Val},\textup{CF},\textup{B}) $,
		\item contains $ \mathfrak{L}(\textup{Val},\textup{CF},\mathbb{Z}) $ or
		\item is equal to $ \mathfrak{L}(\textup{Val},\textup{CF},G) $ for $ G $ an infinite periodic group which is not locally finite.		
	\end{enumerate}
\end{theorem}
\begin{proof}
	Let $ M $ be a monoid. By Corollary \ref{cor: simple}, $ \mathfrak{L}(\textup{Val},\textup{CF},M) =   \mathfrak{L}(\textup{Val}, \textup{CF},N) $ for some simple or 0-simple monoid $ N $. By Fact \ref{fact: simple}, $ N $ either contains a copy of the bicyclic monoid as a submonoid or $ N $ is a group (a group with 0 adjoined). In the former case \textit{(ii)} holds. 
	
	In the latter case, if $ N $ is a group with zero adjoined, then by Proposition \ref{prop: zero} we know that for some group $ G $, $ \mathfrak{L}(\textup{Val}, \textup{CF},N) =\mathfrak{L}(\textup{Val}, \textup{CF},G) $. If $ G $ is not periodic, then it has an element of infinite order which generates a subgroup isomorphic to $ \mathbb{Z} $ and hence \textit{(iii)} follows.
	Otherwise, suppose that $ G $ is locally finite. By Proposition \ref{prop: kambites}, every language in  $ \mathfrak{L}(\textup{Val},\textup{CF},G) $ belongs to  $ \mathfrak{L}(\textup{Val},\textup{CF},H) $ for some finitely generated subgroup $ H $ of $ G $. Since $ G $ is locally finite, $ H $ is finite. Any language $ \mathfrak{L}(\textup{Val},\textup{CF},H) $ is context-free by a result from \cite{Ze11} and hence $ (i) $ holds. 
	The only remaining case is that $ G $ is a periodic group which is not locally finite, in which case $ (iv) $ holds.
\end{proof}

For instance, the result about valence grammars over commutative monoids in \cite{FS02}, now follows as a corollary of Theorem \ref{thm: characterize}.

\begin{corollary}
	Let $ M $ be a commutative monoid. Then $ \mathfrak{L}(\textup{Val},\textup{CF},M) = \mathfrak{L}(\textup{Val},\textup{CF},G) $ for some group $ G $.
\end{corollary}
\begin{proof}
	Since no commutative monoid $ M $ can contain a copy of the bicyclic monoid as a submonoid, the result follows by the proof of Theorem  \ref{thm: characterize}.
\end{proof}

\section{Future work}
Is it possible to prove a pumping lemma for rational monoid automata over permutable monoids?

In Section \ref{sec: pdav}, we prove that a valence PDA over $ M $ is equivalent to a valence automaton over $ \textup{P}_2 \times M$. Can we prove a similar equivalence result for context-free valence grammars?

In Theorem \ref{thm: characterize}, we conclude that when $ M $ is a monoid that contains B, $ \mathfrak{L}(\textup{Val},\textup{CF},M) $ contains the class
$ \mathfrak{L}(\textup{Val},\textup{CF},\textup{B}) $. Since B is not commutative, no correspondence with valence PDA exists, and little is known about the class $ \mathfrak{L}(\textup{Val},\textup{CF},\textup{B}) $, except that it contains the set of partially blind one counter languages. What can we say further about $ \mathfrak{L}(\textup{Val},\textup{CF},\textup{B}) $?

\bibliographystyle{eptcs}
\bibliography{references} 

\end{document}